\documentclass{scrartcl}
\usepackage[numbers]{natbib}
\usepackage{fontenc}
\usepackage{shadethm}
\usepackage{amsthm}
\usepackage{float}
\usepackage{bm}
\usepackage{graphicx}
\usepackage{subfig}
\usepackage{stmaryrd}
\usepackage{mathrsfs} 
\usepackage{amsmath}
\usepackage{amssymb}
\usepackage{wasysym}
\usepackage{sectsty}
\usepackage[hyperfootnotes=false]{hyperref}
\usepackage[a4paper, left=2.7cm, right=2.7cm, top=2.5cm]{geometry}
\usepackage{chngcntr}
\counterwithin*{equation}{section}

\usepackage{cleveref}
\DeclareMathAlphabet{\mathpzc}{OT1}{pzc}{m}{it}

\sectionfont{\normalsize\centering}
\subsectionfont{\footnotesize\centering}

\theoremstyle{plain}
\newtheorem{thm}{Theorem}[section] 

\theoremstyle{definition}
\newtheorem{defn}[thm]{Definition} 
\newtheorem{lem}[thm]{Lemma}

\newtheorem{rem}[thm]{Remark}
\newtheorem{cor}[thm]{Corollary}

\def\XXint#1#2#3{{\setbox0=\hbox{$#1{#2#3}{\int}$ }
		\vcenter{\hbox{$#2#3$ }}\kern-.6\wd0}}

\usepackage[utf8]{inputenc}
\usepackage[german,english,russian]{babel}

\newcounter{MPequ}

\newcounter{AppA}
\newenvironment{AppA}
{\stepcounter{AppA}%
	\addtocounter{equation}{0}%
	\equation}
{\endequation}

\begin{document}\selectlanguage{english}
\begin{center}
\normalsize \textbf{\textsf{Existence of optimal domains for the helicity maximisation problem among domains satisfying a uniform ball condition}}
\end{center}
\begin{center}
	Wadim Gerner\footnote{\textit{E-mail address:} \href{mailto:wadim.gerner@inria.fr}{wadim.gerner@inria.fr}}
\end{center}
\begin{center}
{\footnotesize	Sorbonne Universit\'e, Inria, CNRS, Laboratoire Jacques-Louis Lions (LJLL), Paris, France}
\end{center}
{\small \textbf{Abstract:} 
	In the present work we present a general framework which guarantees the existence of optimal domains for isoperimetric problems within the class of $C^{1,1}$-regular domains satisfying a uniform ball condition as long as the desired objective function satisfies certain properties. We then verify that the helicity isoperimetric problem studied in [J. Cantarella, D. DeTurck, H. Gluck and M. Teytel, J. Math. Phys. 41, 5615 (2000)] satisfies the conditions of our framework and hence establish the existence of optimal domains within the given class of domains. We additionally use the same framework to prove the existence of optimal domains among uniform $C^{1,1}$-domains for a first curl eigenvalue problem which has been studied recently for other classes of domains in [A. Enciso, W. Gerner and D. Peralta-Salas, Trans. Amer. Math. Soc. 377, 4519-4540 (2024)].
\section{Introduction}
In the theory of plasma fusion one measure of stability of the plasma is the helicity of the underlying magnetic field\footnote[2]{See for example arXiv identifier: 2112.01193 which is to appear in Dec 2023 as a chapter in the following book: S. Candelaresi, F. Del Sordo. Stability of plasmas through magnetic helicity. In K. Kuzanyan, N. Yokoi, M. Georgoulis, R. Georgoulis, editors, Helicities in Geophysics, Astrophysics and Beyond. Wiley publication}. It was observed by Woltjer \cite{W58} that in the context of ideal magnetohydrodynamics this so called helicity is a conserved physical quantity. Given a bounded smooth domain $\Omega\subset \mathbb{R}^3$ and a divergence-free field $B$ on $\Omega$ which is tangent to $\partial\Omega$ the helicity of $B$ can be defined as
\begin{equation}
	\label{E1}
	\mathcal{H}(B)(x):=\int_{\Omega}B(x)\cdot \operatorname{BS}(B)(x)d^3x
\end{equation} 
where
\begin{equation}
	\label{E2}
	\operatorname{BS}(B)(x):=\frac{1}{4\pi}\int_{\Omega}B(y)\times \frac{x-y}{|x-y|^3}d^3y
\end{equation}
denotes the Biot-Savart operator. In the context of ideal magnetohydrodynamics $B$ plays the role of the magnetic field. A physical interpretation in terms of linkage of distinct magnetic field lines was obtained in \cite{M69}, \cite{Arnold2014}, \cite{V03}. Then, given an electrically conducting plasma there will be an interplay between the time evolution of the magnetic field $B$ and the flow of the plasma particles due to the Lorentz force. One can then argue in the same spirit as in \cite[Chapter III]{AK98}. Namely, that due to Alfv\'{e}ns theorem \cite{A42}, an underlying plasma fluid is frozen into the magnetic field, in the sense that plasma particles which lie on an initial magnetic field line continue to lie on the same magnetic field line as time passes. Therefore, if the magnetic field is tangent to a domain $\Omega$ and has non-trivially linked field lines, then, due to the fact that distinct field lines cannot cross, also the plasma particles, being frozen into the magnetic field, will form linked structures which cannot be separated. Hence, it follows, in the context of ideal magnetohydrodynamics, where we assume the plasma to be perfectly electrically conducting, that a non-trivial linkage of magnetic field lines leads to a higher (topological) plasma stability and consequently a high helicity, being a measure for the linkage of distinct magnetic field lines, is desirable. In reality, the non-ideal situation, magnetic field line reconnection may occur, see for instance \cite{LPG21}, \cite{B96}, \cite{ELPS17} and references therein. However, obtaining a good understanding of the idealised problem, is a good starting point to gain a better understanding of how one may confine plasmas better within appropriate domains in $3$-space.
Let us point out that helicity as defined in (\ref{E1}) is not scaling invariant, i.e. given a magnetic field
$B$ and a constant $\lambda>0$ we see that $\mathcal{H}(\lambda B)=\lambda^2\mathcal{H}(B)$. However, the field lines of $\lambda B$ coincide with those of $B$ (they are simply traced out with a different speed) so that in fact the linkage of distinct field lines does not change under such a scaling. Therefore, if we wish to regard helicity as a measure for linkage, we should normalise it in the sense that we either fix the magnetic energy
\[
\mathcal{M}(B):=\int_{\Omega}B^2d^3x
\]
or divide by it to achieve a scaling invariant quantity. Both approaches are equivalent and here we adapt the latter.
In view of the above discussion it is of interest to consider the following quantity, where $\Omega\subset\mathbb{R}^3$ is a bounded domain with smooth boundary, more details will be given in section 2,
\begin{equation}
	\label{E3}
	\lambda(\Omega):=\sup_{\substack{B\in L^2(\Omega) \\ \operatorname{div}(B)=0\text{, }B\parallel\partial\Omega \\ \mathcal{M}(B)=1 }}\mathcal{H}(B)=\sup_{\substack{B\in L^2(\Omega)\setminus\{0\} \\ \operatorname{div}(B)=0\text{, }B\parallel\partial\Omega }}\frac{\mathcal{H}(B)}{\mathcal{M}(B)}=\sup_{\substack{B\in L^2(\Omega) \\ \operatorname{div}(B)=0\text{, }B\parallel\partial\Omega\\ \mathcal{H}(B)>0 }}\frac{\mathcal{H}(B)}{\mathcal{M}(B)}
\end{equation}
where the first equality follows from the scaling behaviour and the second equality follows from the fact that helicity is never maximised by magnetic fields of non-positive helicity. The meaning of the tangent boundary condition for an $L^2$-vector field will be explained in section 2.

Our goal then is to try to find a domain $\Omega_*\subset \mathbb{R}^3$ which maximises the quantity $\lambda(\Omega)$ among all other domains $\Omega$. However, once again, by scaling a given domain $\Omega$ by some constant $\lambda>0$, one easily sees that the quantity $\lambda(\Omega)$ may be made arbitrarily large if we do not impose any additional constraint on the allowed domains $\Omega$. It is well-known, and in fact the reason for helicity preservation in ideal magnetohydrodynamics, that helicity is preserved under the action of volume-preserving diffeomorphisms \cite[Theorem A]{CDGT002}, see also \cite[Section 2.3 Corollary]{Arnold2014} and \cite[Lemma 4.5]{G20} for a more abstract manifold setting. In fact, helicity is essentially the only such invariant see \cite{EPT16}, \cite{K16} and also \cite{KPY20}. It is hence natural to restrict attention to domains $\Omega$ of prescribed volume. In fact, it follows from \cite[Theorem B]{CDG00} that
\[
\lambda(\Omega)\leq R(\Omega)
\]
where $R(\Omega)$ is the radius of a ball whose volume is $|\Omega|$, i.e. $R(\Omega)=\sqrt[3]{\frac{3}{4\pi}|\Omega|}$ so that $\lambda(\Omega)$ is uniformly upper bounded among all domains of the same volume.
We thus are interested in the following problem, where $V>0$ is any fixed constant. Find $\Omega_*\subset\mathbb{R}^3$ with $|\Omega_*|=V$ and
\begin{equation}
	\label{E4}
	\lambda(\Omega_*)=\sup_{\substack{\Omega\subset\mathbb{R}^3\\ |\Omega|=V}}\lambda(\Omega)=\sup_{\substack{\Omega\subset\mathbb{R}^3\\ |\Omega|=V}}\frac{1}{\inf_{\substack{B\in L^2(\Omega) \\ \operatorname{div}(B)=0\text{, }B\parallel\partial\Omega\\ \mathcal{H}(B)>0 }}\frac{\mathcal{M}(B)}{\mathcal{H}(B)}}=\frac{1}{\inf_{\substack{\Omega\subset\mathbb{R}^3\\ |\Omega|=V}}\inf_{\substack{B\in L^2(\Omega) \\ \operatorname{div}(B)=0\text{, }B\parallel\partial\Omega\\ \mathcal{H}(B)>0 }}\frac{\mathcal{M}(B)}{\mathcal{H}(B)}}
\end{equation}
where the second equality is a simple reformulation of the original problem for the purpose of transforming the maximisation problem into an appropriate minimisation problem.

The study of problem (\ref{E4}) was initiated in \cite{CDGT002} by Cantarella, DeTurck, Gluck and Teytel. In their work they derived topological constraints for potential optimal domains. Their main result regarding the topology of smooth optimal domains is that any smooth optimal domain, if it exists, must be bounded by tori \cite[Theorem D]{CDGT002}. The existence of optimal domains remained however open.

In the present work we consider an appropriate subclass of domains, namely $C^{1,1}$-regular domains $\Omega\subset\mathbb{R}^3$ satisfying a uniform ball condition, and show that within this class of domains there always exists a domain $\Omega_*$ of prescribed volume which solves (\ref{E4}). Note that, in principle, there might be other smooth domains violating the given uniform ball condition with a larger value for $\lambda(\Omega)$. Hence, the existence of optimal domains in the unrestricted smooth setting remains open.

Let us point out why the class of domains satisfying a uniform ball condition is particularly suited to tackle volume constraint minimisation problems. There are roughly speaking two things that may go wrong regarding compactness properties when dealing with such problems. When we consider a minimising sequence $(\Omega_n)_n$ the first thing that might go wrong is that while each of the domains $\Omega_n$ stays bounded individually, in the limit their diameter might tend to infinity. Hence, these domains might ``approach" an unbounded domain so that we leave the desired class of domains (which are all bounded). The second thing that might go wrong, even if the diameter of all the $\Omega_n$ is uniformly bounded, is the formation of singularities. More precisely, if we wish to optimise an objective function among domains of certain regularity (let's say among domains with a $C^k$-boundary, for some prescribed $k\in \mathbb{N}$) then it might happen that different parts of the boundaries of the $\Omega_n$ approach each other, either from the inside or outside, and ``touch" each other in the limit so that the limiting domain is no longer a manifold with boundary. Of course, even if some form of compactness can be guaranteed with respect to some appropriate topology, it remains to show that the corresponding objective function has nice enough continuity properties in order to exploit variational techniques.

Now, when we deal with volume constraint minimisation problems, any sequence of domains $(\Omega_n)_n$, each of the same volume, satisfying a uniform ball condition will have a uniform bounded diameter. Otherwise one can fit more and more balls inside the $\Omega_n$ each of the same minimal volume $V_{\operatorname{min}}>0$ so that if the diameter becomes too large the volume of $\Omega_n$ would exceed $V$ which is absurd. Hence the first situation as described above cannot occur. We will make this statement precise in section 3. Additionally, the uniform ball condition implies that distinct boundary parts cannot come too close to each other, neither from the inside nor the outside, since otherwise the ball condition would be violated, which rules out the second problem mentioned above. This last intuitive reasoning was made precise in \cite[Theorem 2.8]{GY12}.

The notion of uniform ball domains in the context of shape optimisation was exploited and its compactness properties studied in \cite{GY12}. It is worthwhile to point out that the shape optimisation problem studied in \cite{GY12} is not a volume constraint problem, which shows that the class of domains satisfying a uniform ball condition is rather versatile (for the sake of clarity we mention that in \cite{GY12} the authors considered a class of domains $\Omega$ which are contained in some large ball $B_R$ so that they did not have to deal with the first potential problem regarding a loss of compactness as described above).
\newline
\newline
A second optimisation problem which we would like to address here is a variation of the helicity maximisation problem described above. In this optimisation problem the objective function $\Lambda(\Omega)$ is a modification of $\lambda(\Omega)$ where the supremum in (\ref{E3}) is taken not among all divergence-free fields tangent to the boundary, but in the more restrictive class of divergence-free fields tangent to the boundary which satisfy an additional zero-flux condition. A smooth vector field $X$ on a domain $\Omega$ is said to satisfy the zero-flux condition if for any surface $S\subset \Omega$ with $\partial S\subset \partial\Omega$ the flux of $X$ through $S$ is necessarily zero. The space of zero flux fields may be equivalently expressed as the space of divergence-free fields tangent to the boundary which are $L^2$-orthogonal to all harmonic fields, i.e. fields of vanishing divergence and rotation, see \cite[Hodge decomposition theorem]{CDG02}. One can then similarly ask for the existence of a domain maximising the corresponding quantity $\Lambda(\Omega)$ among all domains $\Omega$ of prescribed volume. This new optimisation problem is in fact equivalent to a shape optimising curl eigenvalue problem, see \cite[Proposition 2.4.3]{G20Diss}, which can also be studied on abstract manifolds. This curl eigenvalue problem is natural from the spectral theoretical point of view because restricting the curl operator to a suitable subset, \cite{YG90}, turns the curl into a self-adjoint operator with compact inverse and hence gives rise to a well-behaved spectrum. While, for example, the shape optimisation problem regarding the first Dirichlet-eigenvalue of the scalar Laplacian has a rich history which can be tracked back to Lord Rayleigh \cite{LR94}\footnote[3]{See also the reprinted version \cite{LR45}.} and was resolved independently by Faber and Krahn in \cite{Fa23} and \cite{Kr25}, the corresponding optimal domain problem for the curl operator has only been investigated recently, initiated independently in \cite[Chapter 2]{G20Diss} and \cite{EPS23}. Letting $\mu_1(\Omega)>0$ denote the smallest positive curl eigenvalue in this context, it follows from \cite[Theorem 2.1 \& Lemma 4.3]{G20}, that we have
\[
\inf_{\substack{B\in L^2(\Omega) \\ \operatorname{div}(B)=0\text{, }B\parallel\partial\Omega\\ \mathcal{H}(B)>0 \\ B \text{ is of zero flux}}}\frac{\mathcal{M}(B)}{\mathcal{H}(B)}=\mu_1(\Omega)
\]
and so according to (\ref{E4}), as pointed out already, maximising $\Lambda(\Omega)$ among domains of fixed volume is the same as minimising the first curl eigenvalue. So from a spectral theoretical point of view restricting the supremum in (\ref{E3}) to zero flux fields is a natural idea.

The minimisation of the first curl eigenvalue among domains of fixed volume has further been investigated in \cite{G23} and \cite{EGPS23}. While the former work derived further geometrical necessary conditions which optimal domains, assuming their existence, must satisfy, the latter dealt with existence questions among two different kinds of classes of domains. First, the existence of optimal domains within the class of uniform $C^{k,\alpha}$ domains for  fixed $k\in \mathbb{N}_{\geq 2}$ and $0<\alpha\leq 1$ contained in a large bounded ball was established, see \cite[Definition 5.1]{EGPS23} for a precise meaning of ``uniformity" in this context. Second, the existence of optimal domains within the class of convex domains was shown \cite[Theorem 1.2]{EGPS23}. For the sake of completeness we point out that the class of divergence-free fields tangent to the boundary and of zero flux coincides, on convex domains, with the class of divergence-free fields tangent to the boundary, i.e. each divergence-free field tangent to the boundary is necessarily a zero flux field \cite[Proposition 2.1]{G20} so that we immediately obtain the following result from \cite[Theorem 1.2]{EGPS23}.
\begin{thm}[\cite{EGPS23}]
	\label{IT1}
	For any given $V>0$ there exists a bounded convex domain $\Omega_*\subset \mathbb{R}^3$ of volume $V$ such that
	\[
	\lambda(\Omega_*)=\sup_{\substack{\Omega\subset\mathbb{R}^3\\|\Omega|=V\\\Omega\text{ convex, bounded}}}\lambda(\Omega).
	\]
	Further, the boundary of any maximiser $\Omega_*$ cannot be analytic.
\end{thm}
To the best of my knowledge no further advances regarding the existence of optimal domains in classes other than the convex domains and uniformly H\"{o}lder domains are known in the literature.

The goal of the present paper is to establish a new class of domains within which the existence of optimal domains can be guaranteed. We note that, as soon as we allow domains to have a non-trivial first de Rham cohomology group, the maximisation problems regarding $\lambda$ and $\Lambda$ are distinct problems, \cite[Proposition 2.1]{G20} and have to be treated separately. In fact, the differences between these problems allowed in \cite[Theorem 1.2]{EPS23} and \cite[Theorem 2.7]{G23} to rule out the possibility for a broad class of rotationally symmetric domains to be optimal for $\Lambda$ while identical arguments do not apply to $\lambda$. Despite some key differences these two optimisation problems also share some key features which we can exploit in order to establish the existence of optimal domains within appropriate classes of domains.
\newline
\newline
The structure of the remaining paper is as follows. In section 2.1 we introduce the necessary notions and notations which will be used throughout the paper. In section 2.2 we state our main results. We start by stating an abstract existence result, which proves the existence of optimal domains within the class of $C^{1,1}$-domains satisfying a uniform ball condition, provided the objective function which we wish to optimise has certain features. We then state the corresponding existence results for our optimisation problems regarding $\lambda(\Omega)$ and $\Lambda(\Omega)$. In section 3 we give the proof of the abstract main result while in section 4 we prove that the objective functions $\lambda$ and $\Lambda$ both satisfy the conditions of the abstract result which will establish the existence of optimal domains in the described class.
\section{Main results}
\subsection{Notation and preliminary notions}
We first introduce the following standard ball conditions and state then a corresponding compactness property.
\begin{defn}[Ball conditions]
	\label{2D1}
	Let $\Omega\subset \mathbb{R}^N$ be an open set.
	\begin{enumerate}
		\item We say that $\Omega$ satisfies the interior ball condition at a given $x_0\in \partial\Omega$ if there is some $y(x_0)\in \mathbb{R}^N$ and $r(x_0)>0$ such that the open ball $B_{r(x_0)}(y(x_0))$ is contained in $\Omega$ and $x_0\in \partial B_{r(x_0)}(y(x_0))$. We say that $\Omega$ satisfies the interior ball condition if it satisfies the interior ball condition at each $x_0\in \partial\Omega$. We say that $\Omega$ satisfies a uniform interior ball condition if there exists some $r_{\Omega}>0$ such that $\Omega$ satisfies the interior ball condition at each $x_0\in \partial\Omega$ and we can choose $r(x_0)\geq r_{\Omega}$.
		\item We say that $\Omega$ satisfies the exterior ball condition at a given $x_0\in \partial\Omega$ if there is some $y(x_0)\in \mathbb{R}^N$ and $r(x_0)>0$ such that the closed ball $\overline{B_{r(x_0)}(y(x_0))}$ is contained in $\mathbb{R}^N\setminus \Omega$ and $x_0\in \partial B_{r(x_0)}(y(x_0))$. We say that $\Omega$ satisfies the exterior ball condition if it satisfies the exterior ball condition at each $x_0\in \partial\Omega$. We say that $\Omega$ satisfies a uniform exterior ball condition if there exists some $r_{\Omega}>0$ such that $\Omega$ satisfies the exterior ball condition at each $x_0\in \partial\Omega$ and we can choose $r(x_0)\geq r_{\Omega}$.
		\item We say that $\Omega$ satisfies the uniform ball condition if $\Omega$ satisfies the uniform interior and uniform exterior ball conditions.
	\end{enumerate}
\end{defn}
One important result regarding the regularity of domains satisfying the uniform ball condition is the following
\begin{thm}[{\cite[Theorem 2.9]{GY12}}]
	\label{2T2}
	Let $\Omega\subset \mathbb{R}^N$ be a bounded open set which satisfies the uniform ball condition, then $\partial\Omega\in C^{1,1}$.
\end{thm}
Before we can talk about compactness we have to introduce a suitable notion of distance. For that matter we recall some standard definitions.
\begin{defn}[Hausdorff distance]
	\label{2D3}
	$\quad$
	\begin{enumerate}
		\item Let $\emptyset\neq K_1,K_2\subset \mathbb{R}^N$ be compact sets. The Hausdorff distance between $K_1$ and $K_2$ is defined by
		\[
		\delta\left(K_1,K_2\right):=\max\left\{\sup_{x\in K_1}\operatorname{dist}(x,K_2),\sup_{x\in K_2}\operatorname{dist}(x,K_1)\right\}.
		\]
		The function $\delta$ defines a metric on the set of non-empty compact subsets of $\mathbb{R}^N$.
		\item Given some $R_0>0$ and open sets $\Omega_1,\Omega_2\subset \overline{B_{R_0}(0)}\subset \mathbb{R}^N$ we define the Hausdorff distance between $\Omega_1$ and $\Omega_2$ relative to $R_0$ by
		\[
		\rho(\Omega_1,\Omega_2)\equiv \rho_{R_0}(\Omega_1,\Omega_2):=\delta\left(\overline{B_{R_0}(0)}\setminus \Omega_1,\overline{B_{R_0}(0)}\setminus \Omega_2\right).
		\]
		This defines a metric on the set of open subsets of $B_{R_0}(0)$.
	\end{enumerate}
\end{defn}
We have the following compactness result
\begin{thm}[{\cite[Theorem 2.8]{GY12}}]
	\label{2T4}
	Let $r_0>0$ and $R_0>2r_0$ be given. Then $\rho_{3R_0}$ turns the set of open sets $\Omega$ which are contained in $B_{R_0}(0)$ and satisfy the uniform ball condition with $r_{\Omega}\geq r_0$ into a compact metric space.
\end{thm}
In view of the regularity result, \Cref{2T2}, we have to make sense of our optimisation problems among the class of $C^{1,1}$-regular domains. Here we will define them more generally on the space of $C^{0,1}$-regular domains. Note that since all convex domains are Lipschitz domains, this in particular allows us to make sense of the corresponding optimisation problems among convex domains.
\newline
\newline
Before we give a precise definition of our objective functions we define two spaces of interest. In the following $\Omega\subset\mathbb{R}^3$ is a bounded $C^{0,1}$-regular open set
\begin{gather}
	\label{2E1}
	\mathcal{V}^{\operatorname{T}}_{\operatorname{div}=0}(\Omega):=\left\{B\in L^2(\Omega,\mathbb{R}^3)\left|\int_{\Omega}B\cdot \operatorname{grad}(f)d^3x=0\right.\text{ for all }f\in H^1(\Omega)\right\},
	\\
	\label{2E2}
	\mathcal{V}^{\operatorname{T},\operatorname{ZF}}_{\operatorname{div}=0}(\Omega):=\left\{B\in \mathcal{V}^{\operatorname{T}}_{\operatorname{div}=0}(\Omega)\left|\int_{\Omega}B\cdot Yd^3x=0\right.\text{ for all }Y\in L^2(\Omega,\mathbb{R}^3),\operatorname{curl}(Y)=0=\operatorname{div}(Y) \right\}
\end{gather}
where $Y$ being curl- and divergence-free is understood in the weak sense. Here T stands for tangent and ZF stands for zero flux because the space introduced in (\ref{2E1}) coincides with the space of divergence-free fields which are tangent to the boundary, while the space introduced in (\ref{2E2}) is the space of divergence-free fields tangent to the boundary which satisfy the zero flux condition. Indeed, if $\Omega$ is smooth enough it follows from the Hodge-decomposition theorem \cite[Corollary 3.5.2]{S95} that each smooth vector field $B\in \mathcal{V}^{\operatorname{T},\operatorname{ZF}}_{\operatorname{div}=0}(\Omega)$ admits a vector potential $A$ which is normal to the boundary. So that if $S\subset \Omega$ is a surface with $\partial S\subset \partial\Omega$ it follows from Stokes theorem that
\[
\int_S B\cdot \mathcal{N}d\sigma=\int_{S}\operatorname{curl}(A)\cdot \mathcal{N}d\sigma=\int_{\partial S}Ad\gamma=0
\]
because $A$ is normal to the boundary. We now define our objective functions, where we recall that we set $\mathcal{M}(B):=\int_{\Omega}B^2d^3x$ to be the magnetic energy of a square integrable vector field and where the helicity $\mathcal{H}(B)$ of any element in $\mathcal{V}^{\operatorname{T}}_{\operatorname{div}=0}(\Omega)$ is defined by the formula (\ref{E1}).
\begin{defn}
	\label{2D5}
	Let $\Omega\subset \mathbb{R}^3$ be a bounded, open set with $C^{0,1}$-boundary. Then we define
	\begin{gather}
		\label{2E3}
		\nu(\Omega):=\inf_{\substack{B\in \mathcal{V}^{\operatorname{T}}_{\operatorname{div}=0}(\Omega)\\ \mathcal{H}(B)>0}}\frac{\mathcal{M}(B)}{\mathcal{H}(B)},
			\\
			\label{2E4}
			\eta(\Omega):=\inf_{\substack{B\in \mathcal{V}^{\operatorname{T},\operatorname{ZF}}_{\operatorname{div}=0}(\Omega)\\\mathcal{H}(B)>0}}\frac{\mathcal{M}(B)}{\mathcal{H}(B)}.
	\end{gather}
\end{defn}
Here, in view of our abstract framework \Cref{2T6}, we adapt the convention to view our optimisation problem as a minimisation problem, c.f. (\ref{E4}).
\newline
\newline
Lastly, we introduce the following three collections of subsets of $\mathbb{R}^N$. Here we let $r_0>0$ and $V\geq \omega_Nr^N_0$ be any fixed constants, where $\omega_N$ denotes the volume of the unit ball
\begin{gather}
	\label{2E5}
	\operatorname{Sub}_c(\mathbb{R}^N):=\{\Omega\subset \mathbb{R}^N\mid\Omega\text{ open, bounded and }\partial\Omega\in C^{1,1}\},
	\\
	\label{2E6}
	S_{r_0}:=\{\Omega\in \operatorname{Sub}_c(\mathbb{R}^N)\mid\Omega\text{ satisfies the uniform ball condition with }r_\Omega\geq r_0\}\text{, }
	\\
	\label{2E7}
	S^V_{r_0}:=\{\Omega\in S_{r_0}\mid|\Omega|=V\},
\end{gather}
where $|\Omega|$ denotes the volume of $\Omega$. We note that the condition $V\geq \omega_Nr^N_0$ is necessary and sufficient to guarantee that $S^V_{r_0}\neq \emptyset$.
\subsection{Statement of results}
We first state our main abstract existence result
\begin{thm}[Abstract framework]
	\label{2T6}
	Let $\mu:\operatorname{Sub}_c(\mathbb{R}^N)\rightarrow (0,\infty)$ be a function with the following properties
	\begin{enumerate}
		\item $\forall \Omega_1,\Omega_2\in \operatorname{Sub}_c(\mathbb{R}^N)$ with $\Omega_1\subseteq \Omega_2$ we have $\mu(\Omega_2)\leq \mu(\Omega_1)$ (``reverse monotonicity").
		\item For every $\Omega\in \operatorname{Sub}_c(\mathbb{R}^N)$ and $X\in C^\infty_c(\mathbb{R}^N,\mathbb{R}^N)$, if $X$ is everywhere outward pointing along $\partial\Omega$ and we let $\psi_t$ denote the global flow of $X$, then $\lim_{t\searrow0}\mu(\psi_t(\Omega))=\mu(\Omega)$. (``outward flow continuity").
		\item There exists a locally bounded function $f:\mathbb{R}_{> 0}\rightarrow \mathbb{R}_{\geq 0}$ and a function $g:\mathbb{R}_{>0}\rightarrow \mathbb{R}_{\geq 0}$ with $g(s)\rightarrow 0$ as $s\rightarrow\infty$ such that for all $\Omega_1,\Omega_2\in \operatorname{Sub}_c(\mathbb{R}^N)$ with $\overline{\Omega}_1\cap \overline{\Omega}_2=\emptyset$ we have $\left|\frac{1}{\mu\left(\Omega_1\cup\Omega_2\right)}-\frac{1}{\min\{\mu(\Omega_1),\mu(\Omega_2)\}}\right|\leq f\left(\max\{|\Omega_1|,|\Omega_2|\}\right)g\left(\operatorname{dist}(\Omega_1,\Omega_2)\right)$. Further suppose that we have $\mu\left(\Omega_1\cup B_r\right)=\min\{\mu(\Omega_1),\mu(B_r)\}$ for every Euclidean ball $B_r$ which has positive distance to $\Omega_1$. (``approximate disjoint minimality")
		\item For every $x\in \mathbb{R}^N$ and for every $\Omega\in \operatorname{Sub}_c(\mathbb{R}^N)$ we have $\mu(x+\Omega)=\mu(\Omega)$. (``translation invariance")
	\end{enumerate}
	Then given any $r_0>0$ and $V\geq \omega_N r^N_0>0$, where $\omega_N$ denotes the volume of the unit ball, the restriction $\mu|_{S^V_{r_0}}$ admits a global minimum.
\end{thm}
In section 4 we will verify that the functions $\nu$ and $\eta$ defined in (\ref{2E3}) and (\ref{2E4}) satisfy conditions (i)-(iv) of \Cref{2T6} respectively which immediately yields the following corollary
\begin{cor}
	\label{2C7}
	Let $r_0>$ and $V\geq \omega_3r^3_0$, then there exist $\Omega_1,\Omega_2\in S^V_{r_0}$ such that
	\[
	\nu(\Omega_1)=\inf_{\Omega\in S^V_{r_0}}\nu(\Omega)\text{ and }\eta(\Omega_2)=\inf_{\Omega\in S^V_{r_0}}\eta(\Omega).
	\]
\end{cor}
\begin{rem}
	It is conjectured in \cite[Section M]{CDGT002} that an optimal shape for the optimisation problem (\ref{E4}) may consist of a solid torus whose major radius equals its minor radius and therefore develops a singularity and that no smooth optimal shape exists. Based on this conjecture one may conjecture that the optimisers for the corresponding uniform ball problem which we prove to exist may consist of solid tori whose major radius $R$ and minor radius $r$ approach each other while maintaining a smallest distance to fit in a ball of radius $r_0$ within the whole of the solid torus, i.e. $R=r+r_0$. However, no numerical experiments have been conducted in the course of the present work to further support this conjecture and it would be interesting to investigate the optimal shapes from a numerical perspective, see also \cite{A-RCRVV18} for numerical methods to compute the first curl eigenvalue on general domains of $\mathbb{R}^3$.
\end{rem}
\section{Proof of the abstract existence result}
Our goal will be to eventually exploit the compactness property \Cref{2T4}. In order to achieve that we have to show that we can find a minimising sequence $(\Omega_n)_n$ of $\mu$ within $S^V_{r_0}$ whose members are all contained in the same bounded set. That is, as explained in the introduction, we have to rule out the possibility that (parts of) our domains run off to infinity.
\begin{lem}
	\label{3L1}
	Given $r_0>0$ let $(\Omega_n)_n\subset S_{r_0}$ be a sequence of connected open sets. If $\operatorname{diam}(\Omega_n)\rightarrow\infty$, then $|\Omega_n|\rightarrow\infty$.
\end{lem}
\begin{proof}[Proof of \Cref{3L1}]
	We assume without loss of generality that $\operatorname{diam}(\Omega_n)>10^nr_0$ for each $n$. We can then find for given $n\in \mathbb{N}$ elements $x,y\in \Omega_n$ with $|x-y|>10^nr_0$ by definition of the diameter. Since $\Omega_n$ is open and connected, it is in particular path connected. So we can connect $x$ and $y$ by a continuous curve $\gamma$, $\gamma(0)=x$, $\gamma(1)=y$. We set $x_1:=x$ and note that by continuity there exists a smallest time $t_2>0$ with $|\gamma(t_2)-x_1|=10r_0$. We set $x_2:=\gamma(t_2)$ and then consider the function $f_2(t):=|\gamma(t)-x_2|$ for $t_2\leq t\leq 1$. We note that $f_2(t_2)=0$ and by the triangle inequality we have $f_2(1)=|x_2-y|\geq |y-x_1|-|x_2-x_1|>(10^n-10)r_0$. Hence, for $n\geq 2$ we have $f_2(1)>20r_0$ and so we can find a smallest $t_2<t_3<1$ with $f_2(t_3)=20r_0$. We let $x_3:=\gamma(t_3)$ and observe that $|x_3-x_1|\geq |x_3-x_2|-|x_2-x_1|=10r_0$ and additionally $|x_3-y|\geq |y-x_2|-|x_2-x_3|>(10^n-30)r_0$ where we used the previously obtained estimate for $|y-x_2|$. In particular, for $n\geq 3$ we see that we found points $x_1,x_2,x_3$ with $|x_i-x_j|\geq 10r_0$ for all $i\neq j$. Thus, for given $n\in \mathbb{N}$ we can repeat the above procedure $n$-times to obtain points $x_1,\dots,x_n\in \Omega_n$ with $|x_i-x_j|\geq 10 r_0$ for all $i\neq j$. It then follows from \cite[Lemma 3.10]{GY12} that for each $x_i$ there exists some $y_i\in \Omega_n$ such that $x_i\in B_{\frac{r_0}{2}}(y_i)\subset \Omega_n$. Since the $x_i$ all have a distance of at least $10r_0$ it follows that the balls $B_{\frac{r_0}{2}}(y_i)$,$i=1,\dots,n$ are disjoint and since they are contained in $\Omega_n$ we thus find $|\Omega_n|\geq \omega_N\left(\frac{r_0}{2}\right)^Nn$ which tends to infinity as $n\rightarrow\infty$.
\end{proof}
We thus obtain the following useful concentration compactness type result for translation invariant objective functions
\begin{cor}
	\label{3C2}
	Let $r_0>0$ and $V\geq \omega_Nr^N_0$ be given. If $\mu:\operatorname{Sub}_c(\mathbb{R}^N)\rightarrow(0,\infty)$ is translation invariant and satisfies the approximate disjoint minimality and reverse monotonicity properties, i.e. satisfies properties (i),(iii) and (iv) of \Cref{2T6}, then there exists a minimising sequence $(\Omega_n)_n\subset S^V_{r_0}$ of $\mu$ and some $R>0$ such that $\Omega_n\subset B_R(0)$ for all $n\in \mathbb{N}$.
\end{cor}
\begin{proof}[Proof of \Cref{3C2}]
	Since $|\Omega_n|=V$ for each $n$ we note that every connected component of each $\Omega_n$ has volume at most $V$. Consequently \Cref{3L1} tells us that there exists some $d>0$ such that the diameter of every connected component of every $\Omega_n$ is at most $d$. Further, we observe that the number of connected components of each $\Omega_n$ is uniformly bounded because by the interior ball property each connected component has volume of at least $\omega_Nr^N_0$. Hence there exists some $m\in \mathbb{N}$ such that $\#\Omega_n\leq m$ for all $n\in \mathbb{N}$.
	\newline
	Let us agree to make use of the convention that $\operatorname{dist}(A,\emptyset)=+\infty$ for any $A\subset \mathbb{R}^N$. We then can fix an arbitrary connected component $C_n$ of each $\Omega_n$. Then we consider $\delta_n:=\operatorname{dist}(C_n,\Omega_n\setminus C_n)$. If $\delta_n\rightarrow \infty$ we define $\Omega_n^1:=C_n$ and $\Omega_n^2:=\Omega_n\setminus C_n$. If, on the other hand, $\delta_n$ is uniformly bounded, we can fix one more additional connected component $\widetilde{C}_n$ of $\Omega_n\setminus C_n$ which realises the distance between $C_n$ and $\Omega_n\setminus C_n$. We then consider $\widehat{C}_n:=C_n\cup \widetilde{C}_n$ and $\widehat{\delta}_n:=\operatorname{dist}(\widehat{C}_n,\Omega_n\setminus \widehat{C}_n)$. Again, if $\widehat{\delta}_n\rightarrow \infty$, we define $\Omega_n^1:=\widehat{C}_n$ and $\Omega_n^2:=\Omega_n\setminus \Omega_n^1$, else we can extract one more connected component from $\Omega_n\setminus \widehat{C}_n$ which realises the distance between $\widehat{C}_n$ and $\Omega_n\setminus \widehat{C}_n$ and add it to $\widehat{C}_n$. Since the number of connected components of $\Omega_n$ is uniformly bounded and the diameters of each connected component are also uniformly bounded, we obtain after finitely many steps a decomposition $\Omega_n=\Omega_n^1\cup \Omega^2_n$ with $\operatorname{diam}(\Omega_n^1)\leq \delta$ for some $\delta>0$ (independent of $n$), $\overline{\Omega}_n^1\cap \overline{\Omega}_n^2=\emptyset$, $\# \Omega_n^2\leq \#\Omega_n -1$ and either $\Omega_n^2=\emptyset$ or else $\operatorname{dist}(\Omega_n^1,\Omega_n^2)\rightarrow\infty$ as $n\rightarrow\infty$. If we can extract a subsequence (denoted in the same way) with $\Omega_n^2=\emptyset$ for all $n$, then we already found a uniformly bounded minimising sequence and the claim follows from the translation invariance of $\mu$. Thus, we are left with considering the situation with $\Omega_n^2\neq \emptyset$ and $\operatorname{dist}(\Omega_n^1,\Omega_n^2)\rightarrow\infty$.
	\newline
	To simplify the notation we define $\lambda(\Omega):=\frac{1}{\mu(\Omega)}$ and $\lambda^+_n:=\max\{\lambda(\Omega_n^1),\lambda(\Omega_n^2)\}$. By assumption $\mu$ satisfies the approximate disjoint minimality property and hence there is a locally bounded function $f$ and a function $g$ with $g(s)\rightarrow 0$ as $s\rightarrow\infty$ such that
	\begin{gather}
		\nonumber
		|\lambda(\Omega_n)-\lambda_n^+|\leq f(\max\{|\Omega_n^1|,|\Omega_n^2|\})g(\operatorname{dist}(\Omega_n^1,\Omega_n^2)).
	\end{gather}
	We note that by the interior ball property we have $|B_{r_0}|\leq \max\{|\Omega_n^1|,|\Omega_n^2|\}\leq V$ and since $f$ is locally bounded, it is in particular bounded on the interval $[|B_{r_0}|,V]$ so that there is a constant $c>0$ (depending on $r_0$ and $V$) such that $|\lambda(\Omega_n)-\lambda^+_n|\leq cg(\operatorname{dist}(\Omega_n^1,\Omega_n^2))$. If $\lambda^+_n=\lambda(\Omega^1_n)$ for infinitely many $n$, we can consider the corresponding sequence $\lambda(\Omega_n^1)$ and observe that since $\operatorname{dist}(\Omega_n^1,\Omega_n^2)\rightarrow\infty$ and $g(s)\rightarrow 0$ as $s\rightarrow\infty$ we have $\lambda_n(\Omega_n^1)-\lambda(\Omega_n)\rightarrow 0$ and $\operatorname{diam}(\Omega^1_n)\leq \delta$ for all $n$. If, on the other hand, $\lambda^+_n=\lambda(\Omega_n^2)$ for all but finitely many $n$, then we can repeat the above procedure. More precisely, we can once more fix some connected component of each $\Omega_n^2$ and consider the distance to the complement within $\Omega_n^2$. If the distance is uniformly bounded, then we can remove a second connected component of each $\Omega_n^2$ and add it to the first so that after finitely many steps we similarly find a decomposition $\Omega^2_n=\Omega^3_n\cup \Omega^4_n$ with $\operatorname{diam}(\Omega^3_n)\leq \delta_2$ for some $\delta_2>0$ independent of $n$, $\overline{\Omega}^3_n\cap \overline{\Omega}^4_n=\emptyset$, $\# \Omega^4_n\leq \Omega^2_n-1\leq \#\Omega_n-2$ and either $\Omega^4_n=\emptyset$ or else $\operatorname{dist}(\Omega^3_n,\Omega^4_n)\rightarrow\infty$. Just like before we find $\lambda(\Omega^2_n)-\max\{\lambda(\Omega^3_n),\lambda(\Omega^4_n)\}\rightarrow 0$ and consequently $\lambda(\Omega_n)-\max\{\lambda(\Omega^3_n),\lambda(\Omega^4_n)\}\rightarrow 0$ as $n\rightarrow\infty$. If $\max\{\lambda(\Omega^3_n),\lambda(\Omega^4_n)\}=\lambda(\Omega^3_n)$ we found a sequence $\lambda(\Omega^3_n)-\lambda(\Omega_n)\rightarrow0$ and $\operatorname{diam}(\Omega^3_n)\leq \delta_2$ for all $n$. If instead again $\max\{\lambda(\Omega^3_n),\lambda(\Omega^4_n)\}=\lambda(\Omega^4_n)$ we can repeat the previous procedure once more with $\Omega^4_n$ in place of $\Omega^2_n$. The main observation now is that the number of connected components of our decomposition which is possibly unbounded reduces in each step by at least one. So since the number of connected components of each $\Omega_n$ is uniformly bounded above, we conclude that after finitely many steps we obtain a decomposition $\Omega_n=\widetilde{\Omega}_n\cup \widehat{\Omega}_n$ with $\overline{\widetilde{\Omega}}_n\cap \overline{\widehat{\Omega}}_n=\emptyset$, $\operatorname{diam}(\widetilde{\Omega}_n)\leq \delta$ for a suitable $\delta>0$ independent of $n$ and $\lambda(\Omega_n)-\lambda(\widetilde{\Omega}_n)\rightarrow0$ as $n\rightarrow\infty$. Due to the translation invariance we can translate the $\widetilde{\Omega}_n$ such that $0\in \widetilde{\Omega}_n$ for all $n$. Finally, if $\widehat{\Omega}_n=\emptyset$, then $|\widetilde{\Omega}_n|=V$. Otherwise $|\widetilde{\Omega}_n|\leq V-|B_{r_0}|$ by the interior ball property. We can now take an Euclidean ball of radius $r_n\geq r_0$ which has a distance of $2 r_0$ to $\widetilde{\Omega}_n$ and such that $|\widetilde{\Omega}_n\cup B_{r_n}|=V$. Then obviously $\widetilde{\Omega}_n\cup B_{r_n}\in S_{r_0}^V$. In addition, by the approximate disjoint minimality property we find $\lambda(\widetilde{\Omega}_n\cup B_{r_n})=\max\{\lambda(\widetilde{\Omega}_n),\lambda(B_{r_n})\}$ because $B_{r_n}$ is an Euclidean ball. We observe that if $\rho>0$ is chosen such that $|B_{\rho}|=V$ and if $\lim_{n\rightarrow\infty}\lambda(\Omega_n)\leq \lambda(B_\rho)$, then because $\lambda(\Omega)=\frac{1}{\mu(\Omega)}$, $B_{\rho}$ itself would be a global minimiser for $\mu$ and so a uniformly bounded minimising sequence exists. Thus we assume now that $\lambda(\Omega_n)>\lambda(B_{\rho})+\epsilon$ for all $n$ and some $\epsilon>0$. Then, since $\lambda(\widetilde{\Omega}_n)=\lambda(\Omega_n)+o(1)$ as $n\rightarrow \infty$, we find $\lambda(\widetilde{\Omega}_n)>\lambda(B_{\rho})$ for large enough $n$. Using the reverse monotonicity property, we find $\lambda(B_{r_n})\leq \lambda(B_\rho)$ because $r_n\leq \rho$ and consequently $\lambda(\widetilde{\Omega}_n\cup B_{r_n})=\lambda(\widetilde{\Omega}_n)=\lambda(\Omega_n)+o(1)$ as $n\rightarrow\infty$ and hence, since $r_n\leq \rho$, we found a new uniformly bounded minimising sequence $\Omega^\prime_n:=\widetilde{\Omega}_n\cup B_{r_n}$ of $\mu$ within $S^V_{r_0}$ (recall once more the relation $\lambda(\Omega)=\frac{1}{\mu(\Omega)}$ and that therefore minimising sequences of $\mu$ correspond to maximising sequences of $\lambda$).
	\end{proof}
\Cref{3C2} enables us to exploit the compactness result \Cref{2T4}.
\begin{proof}[Proof of \Cref{2T6}]
	By \Cref{3C2} we may consider a minimising sequence $(\Omega_n)_n\subset S^V_{r_0}$ such that for a suitable $R>0$ we have $\Omega_n\subset B_R(0)$ for all $n\in \mathbb{N}$. It then follows from \Cref{2T4} that, passing to a subsequence if necessary, the $\Omega_n$ converge to some $\Omega\in S_{r_0}$ with respect to the $\rho_{3R}$ metric. It is immediate from the definition of $\rho$ that since $\Omega_n\subset \overline{B_R(0)}$ for all $n$ we also have $\Omega\subset \overline{B_R(0)}$.
	
	The goal now is to show that $\Omega$ minimises $\mu$.
	
	Since $\Omega$ has a $C^{1,1}$-boundary, \Cref{2T2}, we may consider its outward pointing unit normal $\mathcal{N}$ which is of class $C^{0,1}$. We can then approximate $\mathcal{N}$ by means of the Stone-Weierstrass theorem in $C^0$-norm on $\partial\Omega$ by polynomials, i.e. $C^\infty$-smooth functions. We can then use a bump function to obtain a $C^\infty$-smooth vector field $X\in C^\infty_c(\mathbb{R}^N,\mathbb{R}^N)$ which is everywhere outward pointing along $\partial\Omega$ and which is compactly supported within $B_{2R}(0)$. Fix any such vector field $X$ and let $\psi_t$ denote its flow. We then define for given $m\in \mathbb{N}$ the open sets $V_m:=\psi_{\frac{1}{m}}(\Omega)$ and $U_m:=B_{(3-\frac{1}{m})R}\setminus \overline{V}_m$. We observe that $\overline{U}_m=\overline{B_{(3-\frac{1}{m}R)}(0)}\setminus V_m\subset B_{3R}(0)\setminus \overline{\Omega}$ for every $m$ because $X$ is outward pointing. It therefore follows from the exterior $\Gamma$-property, c.f. \cite[Theorem 2.10]{GY12}, that there exists some sequence $n(m)\in \mathbb{N}$ with $\overline{U}_m\subset B_{3R}(0)\setminus \overline{\Omega}_{n(m)}$ and such that $n(m)\rightarrow\infty$ as $m\rightarrow\infty$. We recall that $\overline{U}_m=\overline{B_{(3-\frac{1}{m})R}(0)}\setminus V_m$ and that $\Omega_n,V_m\subset B_{2R}(0)$ for all $m,n$ (because $X$ is compactly supported within $B_{2R}$) so that
	\[
	\Omega_{n(m)}\subset \overline{\Omega}_{n(m)}=B_{2R}\setminus \left(B_{3R}(0)\setminus \overline{\Omega}_{n(m)}\right)\subset B_{2R}(0)\setminus \left(\overline{B_{(3-\frac{1}{m}R)}(0)}\setminus V_m\right)=V_m=\psi_{\frac{1}{m}}(\Omega).
	\]
	Now the reverse monotonicity principle property of $\mu$ implies
	\[
	\mu\left(\psi_{\frac{1}{m}}(\Omega)\right)\leq \mu(\Omega_{n(m)})\text{ for all }m\in \mathbb{N}.
	\]
	Since $n(m)\rightarrow\infty$ as $m\rightarrow\infty$ and $\Omega_n$ was a minimising sequence, we find
	\[
	\lim_{m\rightarrow\infty}\mu(\Omega_{n(m)})=\inf_{\widetilde{\Omega}\in S^V_{r_0}}\mu(\widetilde{\Omega})
	\]
	while the left hand side of the inequality converges by the outward flow property of $\mu$ to $\mu(\Omega)$. Consequently we find
	\[
	\mu(\Omega)\leq \inf_{\widetilde{\Omega}\in S^V_{r_0}}\mu(\widetilde{\Omega}).
	\]
	We are left with showing that $|\Omega|=V$ because we already know that $\Omega\in S_{r_0}$. First we note that the inclusion $\Omega_{n(m)}\subset \psi_{\frac{1}{m}}(\Omega)$ implies $|\psi_{\frac{1}{m}}(\Omega)|\geq |\Omega_{n(m)}|=V$ for all $m\in \mathbb{N}$ and that a simple change of variables shows that $|\psi_{\frac{1}{m}}(\Omega)|$ converges to $|\Omega|$ so that $|\Omega|\geq V$. The converse inequality is a well-known fact, see \cite[Chapter 6.4 Corollary 1]{DZ11}, which tells us that $|\Omega|\leq \liminf_{n\rightarrow\infty}|\Omega_n|=V$ and consequently $\Omega\in S^V_{r_0}$ as desired.
\end{proof}
\section{Proof of \Cref{2C7}}
In this section we show that the functions $\nu$ and $\eta$ defined in (\ref{2E3}) and (\ref{2E4}) satisfy the conditions of \Cref{2T6}. We recall that according to \cite[Theorem B]{CDG00} there is an absolute constant $c>0$ such that $\frac{c}{\sqrt[3]{|\Omega|}}\leq \min\left\{\nu(\Omega),\eta(\Omega)\right\}$ for all $\Omega\in \operatorname{Sub}_c(\mathbb{R}^3)$ which shows that $\nu$ and $\eta$ map into $(0,\infty)$ and which makes the question of existence of optimal domains in fixed volume classes particularly intriguing.
\newline
\newline
\begin{proof}[Proof of \Cref{2C7}] 
	$\quad$
	\newline
	\newline
	\underline{Property (i):} We want to show that both $\nu$ and $\eta$ have the reverse monotonicity property. So let $\Omega_1\subset \Omega_2$ be bounded, open $C^{1,1}$ sets. Given $B\in \mathcal{V}^{\operatorname{T}}_{\operatorname{div}=0}(\Omega_1)$ (resp. $\mathcal{V}^{\operatorname{T},\operatorname{ZF}}_{\operatorname{div}=0}(\Omega_1)$) we can define $\widetilde{B}:=\chi_{\Omega_1}B\in L^2(\Omega_2,\mathbb{R}^3)$. It follows now straightforward from definitions (\ref{2E1}), (\ref{2E2}) of the spaces $\mathcal{V}^{\operatorname{T}}_{\operatorname{div}=0}(\Omega)$ (resp. $\mathcal{V}^{\operatorname{T},\operatorname{ZF}}_{\operatorname{div}=0}(\Omega)$) that $\widetilde{B}\in \mathcal{V}^{\operatorname{T}}_{\operatorname{div}=0}(\Omega)$ (resp. $\widetilde{B}\in \mathcal{V}^{\operatorname{T},\operatorname{ZF}}_{\operatorname{div}=0}(\Omega)$) and that $\mathcal{M}_{\Omega_2}(\widetilde{B})=\mathcal{M}_{\Omega_1}(B)$, $\mathcal{H}_{\Omega_2}(\widetilde{B})=\mathcal{H}_{\Omega_1}(B)$ (recall that $\mathcal{M}$ denotes the magnetic energy, i.e. the $L^2$-norm squared, and we use a subscript to specify the domain of integration and that $\mathcal{H}$ denotes the helicity (\ref{E1})). Consequently
	\[
	\nu(\Omega_1)=\inf_{\substack{B\in \mathcal{V}^{\operatorname{T}}_{\operatorname{div}=0}(\Omega_1)\\ \mathcal{H}_{\Omega_1}(B)>0}}\frac{\mathcal{M}_{\Omega_1}(B)}{\mathcal{H}_{\Omega_1}(B)}=\inf_{\substack{\widetilde{B}\in \mathcal{V}^{\operatorname{T}}_{\operatorname{div}=0}(\Omega_2)\\ \mathcal{H}_{\Omega_2}(\widetilde{B})>0\\ \widetilde{B}=0\text{ on }\Omega_2\setminus \Omega_1 \\ \widetilde{B}|_{\Omega_1}\in \mathcal{V}^{\operatorname{T}}_{\operatorname{div}=0}(\Omega_1)}}\frac{\mathcal{M}_{\Omega_2}(\widetilde{B})}{\mathcal{H}_{\Omega_2}(\widetilde{B})}\geq \nu(\Omega_2)\text{ (resp. }\eta(\Omega_1)\geq \eta(\Omega_2)\text{)}.
	\]
	\underline{Property (ii):} Let $\Omega\in \operatorname{Sub}_c(\mathbb{R}^3)$ and $X\in C^\infty_c(\mathbb{R}^3,\mathbb{R}^3)$ be such that $X$ is everywhere outward pointing along $\partial\Omega$. Let $\psi_t$ denote the (global) flow of $X$. We have to prove that $\lim_{t\searrow 0}\nu(\psi_t(\Omega))=\nu(\Omega)$ (resp. $\lim_{t\searrow 0}\eta(\psi_t(\Omega))=\eta(\Omega)$). In order to derive this result we set for notational simplicity $\Omega_t:=\psi_t(\Omega)$ and given $B\in \mathcal{V}^{\operatorname{T}}_{\operatorname{div}=0}(\Omega)$ we define a vector field on $\Omega_t$ by, see also \cite[Proof of Lemma 4.4]{EGPS23},
	\[
	B_t(x):=\frac{((\psi_t)_{*}B)(x)}{\det(D\psi_t)(\psi^{-1}_t(x))}\text{, }((\psi_t)_{*}B)(x):=(D\psi_t)\left(\psi^{-1}_t(x)\right)\cdot B\left(\psi^{-1}_t(x)\right).
	\]
	If $X\in L^2(\Omega_t,\mathbb{R}^3)$ is any other arbitrary vector field, applying the change of variables formula yields
	\begin{equation}
		\label{4E1}
		\int_{\Omega_t}B_t(x)\cdot X(x)d^3x=\int_{\Omega}B(x)\cdot \left((D\psi_t)^{\operatorname{Tr}}(x)\cdot X(\psi_t(x))\right)d^3x.
	\end{equation}
	If $X(x)=\nabla f(x)$ for some $f\in H^1(\Omega_t)$, a direct calculation yields $\nabla (f\circ \psi_t)(x)=(D\psi_t)^{\operatorname{Tr}}(x)\cdot (\nabla f)(\psi_t(x))$ and hence (\ref{4E1}) becomes
	\[
	\int_{\Omega_t}B_t(x)\cdot \nabla f(x)d^3x=\int_{\Omega}B(x)\cdot \nabla (f\circ \psi_t)(x)d^3x=0
	\]
	where we used in the last step that $f\circ \psi_t\in H^1(\Omega)$ for every $f\in H^1(\Omega_t)$ and that $B\in \mathcal{V}^{\operatorname{T}}_{\operatorname{div}=0}(\Omega)$. Hence $B_t\in \mathcal{V}^{\operatorname{T}}_{\operatorname{div}=0}(\Omega_t)$. Similarly, if $B\in \mathcal{V}^{\operatorname{T},\operatorname{ZF}}_{\operatorname{div}=0}(\Omega)$, then we observe that if $X$ is curl-free in the weak sense on $\Omega_t$, then $(D\psi_t)^{\operatorname{Tr}}(x)\cdot X(\psi_t(x))$ is curl-free in the weak sense on $\Omega$ (this can be most easily seen by identifying the vector field $X$ with a $1$-form $\omega$ and noting that $(D\psi_t)^{\operatorname{Tr}}(x)\cdot X(\psi_t(x))$ then corresponds to the $1$-form $\psi^\#_t\omega$, keeping in mind that being curl-free corresponds to closedness of the corresponding $1$-form and that pullbacks commute with the exterior derivative). Hence, if $X\in L^2(\Omega_t,\mathbb{R}^3)$ is any curl-free field, then $Y:=(D\psi_t)^{\operatorname{Tr}}(x)\cdot X(\psi_t(x))$ will be also curl-free and we can perform an $L^2$-orthogonal decomposition of $Y$ into $Y=\nabla f+\Gamma$ for a suitable $f\in H^1_0(\Omega)$ and square integrable $\Gamma$ which is div- and curl-free in the weak sense. Then the defining properties of $\mathcal{V}^{\operatorname{T},\operatorname{ZF}}_{\operatorname{div}=0}(\Omega)$ imply that the corresponding integral vanishes and hence $B_t\in \mathcal{V}^{\operatorname{T},\operatorname{ZF}}_{\operatorname{div}=0}(\Omega_t)$. In fact, because $\psi_t$ is a diffeomorphism, the map $B \mapsto B_t$ defines an isomorphism between $\mathcal{V}^{\operatorname{T}}_{\operatorname{div}=0}(\Omega)$ and $\mathcal{V}^{\operatorname{T}}_{\operatorname{div}=0}(\Omega_t)$ (resp. $\mathcal{V}^{\operatorname{T},\operatorname{ZF}}_{\operatorname{div}=0}(\Omega)$ and $\mathcal{V}^{\operatorname{T},\operatorname{ZF}}_{\operatorname{div}=0}(\Omega_t)$). We will first argue that $\mathcal{H}_{\Omega_t}(B_t)=\mathcal{H}_\Omega(B)$ for all $B\in \mathcal{V}^{\operatorname{T}}_{\operatorname{div}=0}(\Omega)$. To see this we observe that if we extend a given $B\in \mathcal{V}^{\operatorname{T}}_{\operatorname{div}=0}(\Omega)$ by zero outside of $\Omega$, we obtain a new vector field $\widetilde{B}\in L^p(\mathbb{R}^3,\mathbb{R}^3)$ for all $1\leq p\leq 2$ and that $\widetilde{B}$ is divergence-free in the weak sense on $\mathbb{R}^3$. It then follows from the Hardy-Littlewood-Sobolev inequality that $\operatorname{BS}_{\mathbb{R}^3}(\widetilde{B})=\operatorname{BS}_{\Omega}(B)\in L^{3}(\mathbb{R}^3,\mathbb{R}^3)$ and it is standard, c.f. \cite[Proposition 1]{CDG01}, that it satisfies $\operatorname{curl}(\operatorname{BS}(\widetilde{B}))=\widetilde{B}\in L^{\frac{3}{2}}(\mathbb{R}^3,\mathbb{R}^3)$ (and $\operatorname{div}(\operatorname{BS}(\widetilde{B}))=0$). To simplify notation we set $Z:=\operatorname{BS}_{\mathbb{R}^3}(\widetilde{B})$ and
	\[
	Z_t(x):=Z^j(\psi^{-1}_t(x))(\partial_i\psi^j_{-t}(x))e_i
	\]
	on $\mathbb{R}^3$ and observe that
	\[
	\operatorname{curl}(Z_t)(x)=\frac{((\psi_t)_{*}\widetilde{B})(x)}{\det(D\psi_t)(\psi^{-1}_t(x))}=:\widetilde{B}_t\text{ on }\mathbb{R}^3.
	\]
	It is clear that since $Z\in L^{3}(\mathbb{R}^3,\mathbb{R}^3)$, that so is $Z_t$ for any $t$. It then follows from \Cref{AC2} that we may use $Z_t$ as a vector potential in order to compute the helicity of $\widetilde{B}_t$. We note that $\widetilde{B}_t=\chi_{\Omega_t}B_t$ and hence
	\[
	\mathcal{H}_{\Omega_t}(B_t)=\mathcal{H}_{\mathbb{R}^3}(\widetilde{B}_t)=\int_{\mathbb{R}^3}\widetilde{B}_t\cdot Z_td^3x=\int_{\Omega_t}B_t(x)\cdot Z_t(x)d^3x.
	\]
Using (\ref{4E1}) with $X=Z_t$ together with the fact that $(D\psi^{-1}_t)(\psi_t(x))D\psi_t(x)=D(\psi_t^{-1}\circ \psi_t)(x)=\operatorname{Id}$ we obtain
\[
\mathcal{H}_{\Omega_t}(B_t)=\int_{\Omega}B(x)\cdot \operatorname{BS}_{\mathbb{R}^3}(\widetilde{B})(x)d^3x=\int_{\Omega}B(x)\cdot \operatorname{BS}_{\Omega}(B)(x)d^3x=\mathcal{H}_{\Omega}(B)
\]
and so indeed helicity is preserved by this isomorphism, see also \cite[Theorem A]{CDGT002} for the case of volume preserving transformations. We have shown so far that the map $B\mapsto B_t$ defines a helicity preserving isomorphism between $\mathcal{V}^{\operatorname{T}}_{\operatorname{div}=0}(\Omega)$ and $\mathcal{V}^{\operatorname{T}}_{\operatorname{div}=0}(\Omega_t)$ (resp. $\mathcal{V}^{\operatorname{T},\operatorname{ZF}}_{\operatorname{div}=0}(\Omega)$ and $\mathcal{V}^{\operatorname{T},\operatorname{ZF}}_{\operatorname{div}=0}(\Omega_t)$). Finally, using once more (\ref{4E1}) we find
\[
\mathcal{M}_{\Omega_t}(B_t)=\int_{\Omega}\frac{|D\psi_t(x)\cdot B(x)|^2}{\det(D\psi_t(x))}d^3x.
\]
It is now easy to see, by means of a Taylor expansion in time, that there is a constant $C>0$ (independent of $B$) such that for all $0\leq t\leq 1$
\[
\mathcal{M}_{\Omega_t}(B_t)\geq \mathcal{M}_{\Omega}(B)\left(1-Ct\right).
\]
It now follows immediately from the definition of $\nu$ (resp. $\eta$) and the above considerations that $\nu(\Omega_t)\geq \nu(\Omega)(1-Ct)$ (resp. $\eta(\Omega_t)\geq \eta(\Omega)(1-Ct)$). On the other hand, since $X$ is outward pointing we have $\Omega\subset \Omega_t$ and thus, by the monotonicity property $\nu(\Omega_t)\leq \nu(\Omega)$ (resp. $\eta(\Omega_t)\leq \eta(\Omega)$) and the claim follows from the sandwich lemma.
\newline
\newline
\underline{Property (iii):} We want to show that there is a locally bounded function $f:\mathbb{R}_{>0}\rightarrow \mathbb{R}_{\geq 0}$ and a function $g:\mathbb{R}_{>0}\rightarrow\mathbb{R}_{\geq 0}$ with $g(s)\rightarrow 0$ as $s\rightarrow\infty$ such that for all $\Omega_1,\Omega_2\in \operatorname{Sub}_c(\mathbb{R}^3)$ with $\overline{\Omega}_1\cap \overline{\Omega}_2=\emptyset$ we have $\left|\frac{1}{\nu(\Omega_1\cup \Omega_2)}-\frac{1}{\min\{\nu(\Omega_1),\nu(\Omega_2)\}}\right|\leq f(\max\{|\Omega_1|,|\Omega_2|\})g(\operatorname{dist}(\Omega_1,\Omega_2))$ (resp. $\left|\frac{1}{\eta(\Omega_1\cup \Omega_2)}-\frac{1}{\min\{\eta(\Omega_1),\eta(\Omega_2)\}}\right|\leq f(\max\{|\Omega_1|,|\Omega_2|\})g(\operatorname{dist}(\Omega_1,\Omega_2))\}$). As already pointed out in the introduction the quantity $\eta(\Omega)$ corresponds to the first curl eigenvalue $\mu_1(\Omega)$ of $\Omega$ and the infimum in the definition of $\eta(\Omega)$ is precisely achieved by the corresponding curl eigenfields \cite[Theorem 2.1]{G20}. Therefore, since the positive spectrum of the curl operator on $\Omega_1\cup \Omega_2$ is simply the union of the corresponding positive spectra of $\Omega_1$ and $\Omega_2$, the smallest positive eigenvalue of $\Omega_1\cup \Omega_2$ is the minimum of the smallest positive eigenvalues of $\Omega_1$ and $\Omega_2$ which yields the claim in this case because we can in fact choose $f=0=g$.

While we can argue similarly for $\nu$ that the infimum in the definition of $\nu$ (\ref{2E3}) is achieved by the eigenfields of a modified Biot-Savart operator \cite[Theorem D \& subsequent comments]{CDG01},\cite[Chapter I Introduction]{CDGT002} and that in this case the eigenfields correspond to the largest positive eigenvalue $\sigma_+>0$ of the compact modified Biot-Savart operator and the quantity $\nu$ is a variational characterisation of $\frac{1}{\sigma_+}$, we cannot argue as in the case of $\eta$ that the spectrum of the union is the union of the spectra of the subdomains because the Biot-Savart operator is a non-local operator. We will now show that $\nu$ nonetheless satisfies the approximate disjoint minimality property. To see this we first fix any two domains $\Omega_1,\Omega_2\in \operatorname{Sub}_c(\mathbb{R}^3)$ of positive distance and define for notational simplicity $\lambda(\Omega):=\frac{1}{\mu(\Omega)}$. We observe that $\lambda(\Omega)=\sup_{B\in \mathcal{V}^{\operatorname{T}}_{\operatorname{div}=0}(\Omega)\setminus\{0\}}\frac{\mathcal{H}(B)}{\mathcal{M}(B)}$, recall \Cref{2D5} and Equation (\ref{E3}). Now, given some $B\in \mathcal{V}^{\operatorname{T}}_{\operatorname{div}=0}(\Omega_1\cup\Omega_2)$ we set $B_1:=B|_{\Omega_1}$ and $B_2:=B|_{\Omega_2}$ and note that $B_i\in \mathcal{V}^{\operatorname{T}}_{\operatorname{div}=0}(\Omega_i)$ for $i=1,2$ which follows easily from the definition. Further, we use a subscript to indicate the domain of integration, so for example $\operatorname{BS}_{\Omega_2}(B_2)(x)=\frac{1}{4\pi}\int_{\Omega_2}B_2(y)\times \frac{x-y}{|x-y|^3}d^3y$ etc.. We then have the following identity, with $\Omega:=\Omega_1\cup \Omega_2$,
\begin{gather}
	\label{4E2}
	\frac{\mathcal{H}_{\Omega}(B)}{\mathcal{M}_{\Omega}(B)}=\frac{\mathcal{H}_{\Omega_1}(B_1)+\mathcal{H}_{\Omega_2}(B_2)}{\mathcal{M}_{\Omega_1}(B_1)+\mathcal{M}_{\Omega_2}(B_2)}+\frac{\int_{\Omega_1}B_1(x)\cdot \operatorname{BS}_{\Omega_2}(B_2)(x)d^3x+\int_{\Omega_2}B_2(x)\cdot \operatorname{BS}_{\Omega_1}(B_1)(x)d^3x}{\mathcal{M}_{\Omega_1}(B_1)+\mathcal{M}_{\Omega_2}(B_2)}.
\end{gather}
We observe that the Biot-Savart operators $\operatorname{BS}_{\Omega_i}(B_i)$ are defined on all of $\mathbb{R}^3$ and so for instance the term involving $\operatorname{BS}_{\Omega_1}(B_1)(x)$ for $x\in \Omega_2$ does not vanish in general. We can now however use the definition of the Biot-Savart operator and the fact that $\delta:=\operatorname{dist}(\Omega_1,\Omega_2)>0$ to estimate
\begin{gather}
	\nonumber
	|\operatorname{BS}_{\Omega_1}(B_1)(x)|\leq \frac{1}{4\pi}\int_{\Omega_1}\frac{|B_1(y)|}{|x-y|^2}d^3y\leq \frac{1}{4\pi\delta^2}\int_{\Omega_1}|B_1(y)|d^3y\text{ for all }x\in \Omega_2.
\end{gather}
Consequently we can estimate
\begin{gather}
	\nonumber
	\left|\int_{\Omega_2}B_2(x)\cdot \operatorname{BS}_{\Omega_1}(B_1)(x)d^3x\right|\leq \frac{1}{4\pi\delta^2}\int_{\Omega_2}\int_{\Omega_1}|B_2(x)||B_1(y)|d^3yd^3x\\
	\nonumber
	\leq \frac{1}{2}\left(\frac{\mathcal{M}_{\Omega_1}(B_1)+\mathcal{M}_{\Omega_2}(B_2)}{4\pi\delta^2}\max\{|\Omega_1|,|\Omega_2|\}\right),
\end{gather}
where we used the elementary inequality $ab\leq \frac{a^2+b^2}{2}$ with $a=|B_1(y)|$,$b=|B_2(x)|$. Letting $R(B):=\frac{\int_{\Omega_1}B_1(x)\cdot \operatorname{BS}_{\Omega_2}(B_2)(x)d^3x+\int_{\Omega_2}B_2(x)\cdot \operatorname{BS}_{\Omega_1}(B_1)(x)d^3x}{\mathcal{M}_{\Omega_1}(B_1)+\mathcal{M}_{\Omega_2}(B_2)}$ we observe that this implies
\begin{gather}
	\label{4E3}
	|R(B)|\leq \frac{\max\{|\Omega_1|,\Omega_2|\}}{4\pi \delta^2}\text{, }\delta=\operatorname{dist}(\Omega_1,\Omega_2).
\end{gather}
It is now standard to verify that $\sup_{B\in \mathcal{V}^{\operatorname{T}}_{\operatorname{div}=0}(\Omega)\setminus\{0\}}\frac{\mathcal{H}_{\Omega_1}(B_1)+\mathcal{H}_{\Omega_2}(B_2)}{\mathcal{M}_{\Omega_1}(B_1)+\mathcal{M}_{\Omega_2}(B_2)}=\max\{\lambda(\Omega_1),\lambda(\Omega_2)\}=:\lambda_+$ and that the supremum is in fact achieved by a maximiser of either $\lambda(\Omega_1)$ or $\lambda(\Omega_2)$ which is set to zero on the respective remaining part of $\Omega$. We can therefore let $B\in \mathcal{V}^{\operatorname{T}}_{\operatorname{div}=0}(\Omega)$ be a maximiser of $\lambda(\Omega)$ and conclude from (\ref{4E2}) and (\ref{4E3})
\begin{gather}
	\nonumber
	\lambda(\Omega)=\frac{\mathcal{H}_{\Omega_1}(B_1)+\mathcal{H}_{\Omega_2}(B_2)}{\mathcal{M}_{\Omega_1}(B_1)+\mathcal{M}_{\Omega_2}(B_2)}+R(B)\leq \lambda_++\frac{\max\{|\Omega_1|,|\Omega_2|\}}{4\pi\delta^2}\Rightarrow \lambda(\Omega)-\lambda_+\leq \frac{\max\{|\Omega_1|,|\Omega_2|\}}{4\pi\delta^2}.
\end{gather}
Conversely, if we let $\widetilde{B}\in \mathcal{V}^{\operatorname{T}}_{\operatorname{div}=0}(\Omega)$ be a maximiser of $\frac{\mathcal{H}_{\Omega_1}(B_1)+\mathcal{H}_{\Omega_2}(B_2)}{\mathcal{M}_{\Omega_1}(B_1)+\mathcal{M}_{\Omega_2}(B_2)}$, then we find once more by means of (\ref{4E2}) and (\ref{4E3})
\begin{gather}
	\nonumber
	\lambda_+=\frac{\mathcal{H}_{\Omega}(\widetilde{B})}{\mathcal{M}_{\Omega}(\widetilde{B})}-R(\widetilde{B})\leq \lambda(\Omega)+\frac{\max\{|\Omega_1|,|\Omega_2|\}}{4\pi\delta^2}\Rightarrow \lambda_+-\lambda(\Omega)\leq \frac{\max\{|\Omega_1|,|\Omega_2|\}}{4\pi\delta^2}.
\end{gather}
Overall we conclude that $|\lambda(\Omega_1\cup\Omega_2)-\max\{\lambda(\Omega_1),\lambda(\Omega_2)\}|\leq \frac{\max\{|\Omega_1|,|\Omega_2|\}}{4\pi(\operatorname{dist}(\Omega_1,\Omega_2))^2}$ so that we may let $f(s):=\frac{s}{4\pi}$ and $g(s):=\frac{1}{s^2}$ which proves the desired inequality for $\nu$ upon recalling that $\lambda(\Omega)=\frac{1}{\nu(\Omega)}$.

We are left with proving that $\lambda(\Omega_1\cup B_r)=\max\{\lambda(\Omega_1),\lambda(B_r)\}$ for every $\Omega_1\in \operatorname{Sub}_c(\mathbb{R}^3)$ and any Euclidean ball $B_r$ with positive distance to $\Omega_1$. This will however follow from our previous arguments once we show that $R(B)=0$ for all $B\in \mathcal{V}^{\operatorname{T}}_{\operatorname{div}=0}(\Omega_1\cup B_r)$ since our previous arguments in that case show that $\lambda(\Omega)-\lambda_+\leq 0$ and $\lambda_+-\lambda(\Omega)\leq 0$ where $\Omega:=\Omega_1\cup B_r$ and $\lambda_+:=\max\{\lambda(\Omega_1),\lambda(B_r)\}$. With our previous notation, where we let $\Omega_2:=B_r$, we observe that
\begin{gather}
	\nonumber
	\int_{\Omega_1}B_1(x)\cdot \operatorname{BS}_{\Omega_2}(B_2)(x)d^3x=\frac{1}{4\pi}\int_{\Omega_1}\int_{\Omega_2}B_1(x)\cdot \left(B_2(y)\times \frac{x-y}{|x-y|^3}\right)d^3yd^3x
	\\
	\nonumber
	=\int_{\Omega_2}B_2(y)\cdot \operatorname{BS}_{\Omega_1}(B_1)(y)d^3y
\end{gather}
where we used the cyclic property of the scalar product. It is therefore enough to show that $\int_{B_r}B_2(y)\cdot \operatorname{BS}_{\Omega_1}(B_1)(y)d^3y=0$ for all $B_i\in \mathcal{V}^{\operatorname{T}}_{\operatorname{div}=0}(\Omega_i)$, $i=1,2$. Because $B_r$ is a ball and each $B_2\in \mathcal{V}^{\operatorname{T}}_{\operatorname{div}=0}(B_r)$ is div-free and tangent to the boundary, it follows from the Hodge-decomposition theorem, \cite[Corollary 3.5.2 \& Theorem 2.6.1]{S95} that there exists an $H^1$-vector field $A$ on $B_r$ which is normal to the boundary with $B_2=\operatorname{curl}(A)$. Consequently
\begin{gather}
	\nonumber
	\int_{\Omega_2}B_2(y)\cdot \operatorname{BS}_{\Omega_1}(B_1)(y)d^3y=\int_{B_r}\operatorname{curl}(A)(y)\cdot \operatorname{BS}_{\Omega_1}(B_1)(y)d^3y=\int_{B_r}A(y)\cdot \operatorname{curl}(\operatorname{BS}_{\Omega_1}(B_1)(y))d^3y=0,
\end{gather}
where we used the fact that $A$ is normal to the boundary and thus no boundary terms appear when integrating by parts and the fact that $\operatorname{curl}(\operatorname{BS}_{\Omega_1}(B_1)(y))=0$ for all $y\in B_r$ because $B_r$ and $\Omega_1$ are disjoint, see \cite[Proposition 1]{CDG01}. As explained, this shows that $\lambda(\Omega_1\cup B_r)=\max\{\lambda(\Omega_1),\lambda(B_r)\}$ for every Euclidean ball $B_r$ which has positive distance to $\Omega_1$. Recalling the relation $\lambda(\Omega)=\frac{1}{\nu(\Omega)}$ we conclude that $\nu$ satisfies the approximate disjoint minimality property.
\newline
\newline
\underline{Property (iv):} The remaining argument is straightforward. We can similarly as in part (ii), by replacing $\psi_t$ by the diffeomorphism $\psi$ which induces a translation by $x$, define an isomorphism between $\mathcal{V}^{\operatorname{T}}_{\operatorname{div}=0}(\Omega)$ and $\mathcal{V}^{\operatorname{T}}_{\operatorname{div}=0}(\psi(\Omega))$ (resp. $\mathcal{V}^{\operatorname{T},\operatorname{ZF}}_{\operatorname{div}=0}(\Omega)$ and $\mathcal{V}^{\operatorname{T},\operatorname{ZF}}_{\operatorname{div}=0}(\psi(\Omega))$). As we have seen this isomorphism preserves helicity and because $\psi$ induces an isometry, more precisely we even have $D\psi(y)=\operatorname{Id}$ for all $y\in \Omega$, one easily infers that the defined isomorphism also preserves the magnetic energy $\mathcal{M}$. From this and the definition of $\nu$ (resp. $\eta$) it immediately follows that we have for all $x\in \mathbb{R}^3$ and all $\Omega\in \operatorname{Sub}_c(\mathbb{R}^3)$ $\nu(x+\Omega)=\nu(\Omega)$ (resp. $(\eta(x+\Omega)=\eta(\Omega))$).
\newline
\newline
We conclude that $\nu$ as well as $\eta$ satisfy properties (i)-(iv) of \Cref{2T6} and thus \Cref{2C7} follows.
\end{proof}
\section*{Acknowledgements}
This work has been supported by the Inria AEX StellaCage. Further, the author would like to thank an anonymous referee for pointing out the non-locality of the Biot-Savart operator which had not been taken into account correctly in the first version of this manuscript.
\section*{Data availability}
Data sharing is not applicable to this article as no new data were created or analysed in this study.
\section*{Conflict of interest}
The author has no conflict of interest to declare.
\appendix
\renewcommand{\thesection}{A}
\section{Gauge invariance of helicity}
We first prove the following approximation lemma, c.f. \cite[Lemma 5.3.11]{G20Diss}
\begin{lem}[Approximation lemma]
	\label{AL1}
	Let $1<p<3$ and let $q$ be given by the equation $\frac{1}{q}=\frac{1}{p}-\frac{1}{3}$. If $A\in L^q(\mathbb{R}^3,\mathbb{R}^3)$ with $\operatorname{curl}(A)\in L^p(\mathbb{R}^3,\mathbb{R}^3)$, then there exists a sequence $(A_n)_n\subset C^{\infty}_c(\mathbb{R}^3,\mathbb{R}^3)$ such that
	\begin{enumerate}
		\item $A_n\rightarrow A$ in $L^q(\mathbb{R}^3,\mathbb{R}^3)$,
		\item $\operatorname{curl}(A_n)\rightarrow \operatorname{curl}(A)$ in $L^p(\mathbb{R}^3,\mathbb{R}^3)$.
	\end{enumerate}
\end{lem}
\begin{proof}[Proof of \Cref{AL1}]
	First, we can fix a sequence of bump functions $(\rho_n)_n\subset C^\infty_c(\mathbb{R}^3)$ with $0\leq \rho_n\leq 1$, $\rho_n(x)=1$ for all $|x|\leq n$, $\rho_n(x)=0$ for all $|x|\geq 3n$ and $|\operatorname{grad}(\rho_n)(x)|\leq\frac{C}{n}$ for a constant $C>0$ independent of $n$. In addition, let $(\phi_n)_n$ be a standard Dirac sequence. We then define $A_n(x):=(\phi_n\star (\rho_n A))(x)$, where $\star$ denotes the convolution operator. It is standard that $A_n\rightarrow A$ in $L^q(\mathbb{R}^3,\mathbb{R}^3)$. We further compute by properties of the convolution operator (we make use of the Einstein summation convention)
	\begin{gather}
		\nonumber
		\operatorname{curl}(A_n)(x)=\epsilon^{ijk}e_k\int_{\mathbb{R}^3}(\partial_i\phi_n)(x-y)\rho_n(y)A_j(y)d^3y
	\end{gather}
				\begin{AppA}
		\label{AE1}
		=\epsilon^{ijk}e_k\int_{\mathbb{R}^3}(\partial_i\rho_n)(y)\phi_n(x-y)A_j(y)d^3y-\epsilon^{ijk}e_k\int_{\mathbb{R}^3}\partial_i(\phi_n(x-y)\rho_n(y))A_j(y)d^3y.
	\end{AppA}
	We note first that by the convolution inequality and normalisation of the Dirac sequence
	\begin{gather}
		\nonumber
		\|\phi_n\star ((\partial_i\rho_n) A_j)\|_{L^p(\mathbb{R}^3)}\leq \|\phi_n\|_{L^1(\mathbb{R}^3)}\|(\partial_i\rho_n)A_j\|_{L^p(\mathbb{R}^3)}=\|(\partial_i\rho_n)A_j\|_{L^p(\mathbb{R}^3)}.
	\end{gather}
	In addition, $\partial_i\rho_n(x)=0$ for $|x|\leq n$ because $\rho_n$ is constant on this set. We recall that $\frac{1}{q}+\frac{1}{3}=\frac{1}{p}$ and so the generalised H\"{o}lder inequality implies
	\begin{gather}
		\nonumber
		\|(\partial_i\rho_n)A_j\|_{L^p(\mathbb{R}^3)}=\|(\partial_i\rho_n)A_j\|_{L^p(\mathbb{R}^3\setminus B_n(0))}\leq \|\partial_i\rho_n\|_{L^3(\mathbb{R}^3)}\|A_j\|_{L^q(\mathbb{R}^3\setminus B_n(0))}.
	\end{gather}
	We observe that $\|A_j\|_{L^q(\mathbb{R}^3\setminus B_n(0))}$ converges to zero by dominated convergence, since $A_j\in L^q(\mathbb{R}^3)$ and that $\|\partial_i\rho_n\|_{L^3(\mathbb{R}^3)}$ is uniformly bounded due to the fact that each $\rho_n$ is supported within $B_{3n}(0)$ and due to the decay of the derivatives of the $\rho_n$. This implies that the first term in (\ref{AE1}) converges to zero in $L^p(\mathbb{R}^3)$. As for the second term in (\ref{AE1}) we note that the existence of $\operatorname{curl}(A)$ implies that
	\begin{gather}
		\nonumber
		-\epsilon^{ijk}e_k\int_{\mathbb{R}^3}\partial_i(\phi_n(x-y)\rho_n(y))A_j(y)d^3y=\int_{\mathbb{R}^3}\phi_n(x-y)\rho(y)\operatorname{curl}(A)(y)d^3y=(\phi_n\star(\rho_n \operatorname{curl}(A)))(x).
	\end{gather}
	Hence, the second term in (\ref{AE1}) converges to $\operatorname{curl}(A)$ in $L^p(\mathbb{R}^3,\mathbb{R}^3)$ because $\operatorname{curl}(A)\in L^p(\mathbb{R}^3,\mathbb{R}^3)$ by assumption. This concludes the proof.
\end{proof}
We are now in the position to prove the gauge independence of helicity, c.f. \cite[Corollary 5.5.1]{G20Diss}
\begin{cor}[Gauge invariance of helicity]
	\label{AC2}
	Let $A_1,A_2\in L^3(\mathbb{R}^3,\mathbb{R}^3)$ with $\operatorname{curl}(A_1)=\operatorname{curl}(A_2)\in L^{\frac{3}{2}}(\mathbb{R}^3,\mathbb{R}^3)$. Then
	\begin{gather}
		\nonumber
		\int_{\mathbb{R}^3}A_1\cdot \operatorname{curl}(A_1)d^3x=\int_{\mathbb{R}^3}A_2\cdot \operatorname{curl}(A_2)d^3x.
	\end{gather}
\end{cor}
\begin{proof}[Proof of \Cref{AC2}]
	For notational simplicity we define $B:=\operatorname{curl}(A_1)$ and $A:=A_1-A_2$, then the claimed identity is equivalent to $\int_{\mathbb{R}^3}A\cdot Bd^3x=0$. To see this we can apply \Cref{AL1} with $q=3$ and $p=\frac{3}{2}$ and find a sequence $(A_n)_n$ of smooth and compactly supported vector fields, such that $A_n\rightarrow A$ in $L^3(\mathbb{R}^3,\mathbb{R}^3)$ and $\operatorname{curl}(A_n)\rightarrow \operatorname{curl}(A)=0$ in $L^{\frac{3}{2}}(\mathbb{R}^3,\mathbb{R}^3)$. Then, since $3$ is the H\"{o}lder conjugate of $\frac{3}{2}$ we obtain the following identities
	\begin{gather}
		\nonumber
		\int_{\mathbb{R}^3}B\cdot Ad^3x=\lim_{n\rightarrow\infty}\int_{\mathbb{R}^3}\operatorname{curl}(A_1)\cdot A_nd^3x=\lim_{n\rightarrow\infty}\int_{\mathbb{R}^3}\operatorname{curl}(A_n)\cdot A_1d^3x=0,
	\end{gather}
	where we used the defining properties of the curl operator.
\end{proof}
\bibliographystyle{plain}
\bibliography{mybibfileNOHYPERLINK}

\begin{thebibliography}{10}

\bibitem{A42}
H.~Alfv\'{e}n.
\newblock On the existence of electromagnetic-hydrodynamic waves.
\newblock {\em Arkiv f\"or Matematik, Astronomi och Fysik.}, 29B(2):1--7, 1942.

\bibitem{A-RCRVV18}
A.~Alonso-Rodr\'{\i}guez, J.~Cama{\~{n}}o, R.~Rodr\'{\i}guez, A.~Valli, and
  P.~Venegas.
\newblock Finite {E}lement {A}pproximation of the {S}pectrum of the {C}url
  {O}perator in a {M}ultiply {C}onnected {D}omain.
\newblock {\em Found. Comput. Math.}, 18:1493--1533, 2018.

\bibitem{Arnold2014}
V.~I. Arnold.
\newblock The asymptotic {H}opf invariant and its applications.
\newblock In Alexander~B. Givental, Boris~A. Khesin, Alexander~N. Varchenko,
  Victor~A. Vassiliev, and Oleg~Ya. Viro, editors, {\em Vladimir I. Arnold -
  Collected Works: Hydrodynamics, Bifurcation Theory, and Algebraic Geometry
  1965-1972}, pages 357--375. Springer Berlin Heidelberg, Berlin, Heidelberg,
  2014.

\bibitem{AK98}
V.I. Arnold and B.A. Khesin.
\newblock {\em Topological Methods in Hydrodynamics.}
\newblock Springer Verlag, 1998.

\bibitem{B96}
D.~Biskamp.
\newblock Magnetic reconnection in plasmas.
\newblock {\em Astrophysics and Space Science}, 242:165--207, 1996.

\bibitem{CDGT002}
J.~Cantarella, DeTurck D., H.~Gluck, and M.~Teytel.
\newblock Isoperimetric problems for the helicity of vector fields and the
  {B}iot-{S}avart and curl operators.
\newblock {\em Journal of Mathematical Physics.}, 41:5615--5641, 2000.

\bibitem{CDG00}
J.~Cantarella, D.~DeTurck, and H.~Gluck.
\newblock Upper bounds for the writhing of knots and the helicity of vector
  fields.
\newblock In J.~Gilman, X.-S. Lin, and W.~Menasco, editors, {\em Proceedings of
  the Conference in Honor of the 70th Birthday of Joan Birman, AMS/IP Series on
  Advanced Mathematics}. International Press, 2000.

\bibitem{CDG01}
J.~Cantarella, D.~DeTurck, and H.~Gluck.
\newblock The {B}iot-{S}avart operator for application in knot theory, fluid
  dynamics, and plasma physics.
\newblock {\em Journal of Mathematical Physics}, 42(2):876--905, 2001.

\bibitem{CDG02}
J.~Cantarella, D.~DeTurck, and H.~Gluck.
\newblock Vector calculus and the topology of domains in 3-space.
\newblock {\em American Mathematical Monthly}, 109(5):409--442, 2002.

\bibitem{DZ11}
M.C. Delfour and J.-P. Zolésio.
\newblock {\em Shapes and {G}eometries}.
\newblock {SIAM}, second edition, 2011.

\bibitem{EGPS23}
A.~Enciso, W.~Gerner, and D.~Peralta-Salas.
\newblock Optimal convex domains for the first curl eigenvalue in dimension
  three.
\newblock {\em Trans. Amer. Math. Soc.}, 377:4519--4540, 2024.

\bibitem{ELPS17}
A.~Enciso, R.~Lucà, and D.~Peralta-Salas.
\newblock Vortex reconnection in the three dimensional {N}avier-{S}tokes
  equations.
\newblock {\em Advances in Mathematics}, 309:452--486, 2017.

\bibitem{EPS23}
A.~Enciso and D.~Peralta-Salas.
\newblock Non-existence of axisymmetric optimal domains with smooth boundary
  for the first curl eigenvalue.
\newblock {\em Ann. Sc. Norm. Sup. Pisa Cl. Sci.}, XXIV:311--327, 2023.

\bibitem{EPT16}
A.~Enciso, D.~Peralta-Salas, and F.~Torres~de Lizaur.
\newblock Helicity is the only integral invariant of volume-preserving
  transformations.
\newblock {\em Proc. Natl. Acad. Sci. USA}, 113(8):2035--40, 2016.

\bibitem{Fa23}
G.~Faber.
\newblock {B}eweis, dass unter allen homogenen {M}embranen von gleicher
  {F}l\"{a}che und gleicher {S}pannung die kreisf\"{o}rmige den tiefsten
  {G}rundton gibt.
\newblock {\em Sitzungsber. Math.-Phys. Kl. Bayer. Akad. Wiss. {M}\"{u}nchen},
  pages 169--172, 1923.

\bibitem{G20}
W.~Gerner.
\newblock {E}xistence and characterisation of magnetic energy minimisers on
  oriented, compact {R}iemannian 3-manifolds with boundary in arbitrary
  helicity classes.
\newblock {\em Ann. Global Anal. Geom.}, 58:267 -- 285, 2020.

\bibitem{G20Diss}
W.~Gerner.
\newblock {\em \href{https://doi.org/10.18154/RWTH-2020-12414}{{M}inimisation
  {P}roblems in {I}deal {M}agnetohydrodynamics}}.
\newblock PhD thesis, RWTH Aachen University, 2020.

\bibitem{G23}
W.~Gerner.
\newblock Isoperimetric problem for the first curl eigenvalue.
\newblock {\em J. Math. Anal. Appl.}, 519:126808, 2023.

\bibitem{GY12}
B.-Z. Guo and D.-H. Yang.
\newblock Some compact classes of open sets under {H}ausdorff distance and
  application to shape optimization.
\newblock {\em {SIAM} J. Control Optim.}, 50(1):222--242, 2012.

\bibitem{KPY20}
B.~Khesin, D.~Peralta-Salas, and C.~Yang.
\newblock A {B}asis of {C}asimirs in 3{D} {M}agnetohydrodynamics.
\newblock {\em Int. Math. Res. Not.}, 2021:13645--13660, 2020.

\bibitem{Kr25}
E.~Krahn.
\newblock \"{U}ber eine von {R}ayleigh formulierte {M}inimaleigenschaften des
  {K}reises.
\newblock {\em Math. Ann.}, 94:97--100, 1925.

\bibitem{K16}
E.A. Kudryavtseva.
\newblock Helicity is the only invariant of incompressible flows whose
  derivative is continuous in the ${C}^1$ topology.
\newblock {\em Math. Notes}, 99:611--615, 2016.

\bibitem{LPG21}
T.~Li, E.~Priest, and R.~Guo.
\newblock Three-dimensional magnetic reconnection in astrophysical plasmas.
\newblock {\em Proc. R. Soc. A}, 477:20200949, 2021.

\bibitem{M69}
K.~Moffatt.
\newblock The degree of knottedness of tangled vortex lines.
\newblock {\em Journal of Fluid Mechanics}, 35:117--129, 1969.

\bibitem{LR94}
J.W.S. Rayleigh.
\newblock {\em {T}he {T}heory of {S}ound}.
\newblock London 1894\slash 1896, second edition, 1894\slash 1896.

\bibitem{LR45}
J.W.S. Rayleigh.
\newblock {\em {T}he {T}heory of {S}ound}.
\newblock Dover {B}ooks on {P}hysics, reprint edition, 1945.

\bibitem{S95}
G.~Schwarz.
\newblock {\em Hodge Decomposition - A Method for Solving Boundary Value
  Problems.}
\newblock Springer Verlag, 1995.

\bibitem{V03}
T.~Vogel.
\newblock On the asymptotic linking number.
\newblock In {\em Proc. Amer. Math. Soc.}, volume 131, pages 2289--2297, 2003.

\bibitem{W58}
L.~Woltjer.
\newblock A theorem on force-free magnetic fields.
\newblock In {\em Proc. Natl. Acad. Sci. USA}, volume~44, pages 489--491, 1958.

\bibitem{YG90}
Z.~Yoshida and Y.~Giga.
\newblock Remarks on spectra of operator rot.
\newblock {\em Mathematische Zeitschrift}, 204(1):235--245, 1990.

\end{thebibliography}
\footnotesize
\end{document}